\documentclass{article} 
\usepackage{fourier}
\usepackage{amsmath,amssymb,amsthm}
\usepackage{fullpage}
\usepackage{url}
\usepackage[numbers]{natbib}
\usepackage{color}

\usepackage{algorithm,algorithmic}
\usepackage[colorlinks=true,linkcolor=blue]{hyperref}
\usepackage{url, color}

\newtheorem{theorem}{Theorem}
\newtheorem{lemma}[theorem]{Lemma}

\theoremstyle{definition}
\newtheorem{definition}{Definition}

\let\inf\undef
\DeclareMathOperator*{\inf}{\vphantom{p}inf}
\let\sup\undef
\DeclareMathOperator*{\sup}{\vphantom{p}sup}

\newcommand{\mrm}[1]{\mathrm{#1}}

\newcommand{\p}{\ensuremath{\mathbf{p}}}

\newcommand{\img}{{\sf Img}}

\newcommand{\pred}{\widehat{y}}

\newcommand{\argmin}[1]{\underset{#1}{\mrm{argmin}} \ }

\newcommand{\reals}{\mathbb{R}}
\newcommand{\En}{\mathbb{E}}  

\newcommand{\inner}[1]{\left\langle #1 \right\rangle}

\newcommand{\ind}[1]{{\bf 1}\left\{#1\right\}}
\newcommand{\tr}{\ensuremath{{\scriptscriptstyle\mathsf{T}}}}

\newcommand\e{\mathbf{e}}
\newcommand\s{\mathbf{s}}
\newcommand\w{\mathbf{w}}
\newcommand\x{\mathbf{x}}
\newcommand\y{\mathbf{y}}
\newcommand\z{\mathbf{z}}

\renewcommand\v{\mathbf{v}}

\newcommand\cD{\mathcal{D}}
\newcommand\K{\mathcal{K}}

\newcommand\cN{\mathcal{N}}

\newcommand\X{\mathcal{X}}

\newcommand\Y{\mathcal{Y}}
\newcommand\Z{\mathcal{Z}}
\newcommand\F{\mathcal{F}}
\newcommand\G{\mathcal{G}}

\newcommand\ldim{\mathrm{Ldim}}
\newcommand\fat{\mathrm{fat}}

\newcommand\loss{\boldsymbol{\ell}}

\def\deq{\triangleq}

\newcommand{\bmu}{\ensuremath{\boldsymbol{\mu}}}

\include{multiminimax}

\newcommand{\thresh}[1]{\tau_{\delta}(#1)}

\begin{document}

\title{Sequential Probability Assignment with Binary Alphabets and Large Classes of Experts}

\author{Alexander Rakhlin\\ University of Pennsylvania
\and
Karthik Sridharan\\ Cornell University
}

\maketitle

\begin{abstract}
	We analyze the problem of sequential probability assignment for binary outcomes with side information and logarithmic loss, where regret---or, redundancy---is measured with respect to a (possibly infinite) class of experts. We provide upper and lower bounds for minimax regret in terms of sequential complexities of the class, introduced in \cite{RakSriTew14ptrf,RakSriTew10}. These complexities were recently shown to give matching (up to logarithmic factors) upper and lower bounds for sequential prediction with general convex Lipschitz loss functions \cite{RakSri14a,RakSri15nonparametric}.
To deal with unbounded gradients of the logarithmic loss, we present a new analysis  that employs a sequential chaining technique with a Bernstein-type bound. 
The introduced complexities are intrinsic to the problem of sequential probability assignment, as illustrated by our lower bound.
	
We also consider an example of a large class of experts parametrized by vectors in a high-dimensional Euclidean ball (or a Hilbert ball). The typical discretization approach fails, while our techniques give a non-trivial bound. For this problem we also present an algorithm based on regularization with a self-concordant barrier. This algorithm is of an independent interest, as it requires a bound on the function values rather than gradients.
\end{abstract}

\section{Introduction}
	
	In this paper we study the problem of sequential prediction of a string of bits $ (y_1,\ldots,y_n) \deq y_{1:n} \in\{0,1\}^n$. At each round $t=1,\ldots,n$, the forecaster  observes side information $x_t\in\X_t$, decides on the probability $\pred_t\in[0,1]$ of the event $y_t=1$, observes the outcome $y_t\in\{0,1\}$, and pays according to the logarithmic (or, \emph{self-information}) loss function
	$$\loss(\pred_t,y_t) = -\ind{y=1}\log\pred_t - \ind{y_t=0}\log(1-\pred_t).$$
At each time instance $t$, the side-information set $\X_t$ is a subset of an abstract set $\X$. The subset $\X_t$ is allowed to depend on the history $h_{1:t-1}\deq (x_{1:t-1},y_{1:t-1})$, and the functions $\X_t:(\X\times\Y)^{t-1}\to 2^\X$ are assumed to be known to the forecaster.

The goal of the forecaster is to predict as well as a benchmark set $\F$ of functions---sometimes called ``experts''---mapping $\X$ to $[0,1]$.  More specifically, the goal is to keep \emph{regret}
	$$\sum_{t=1}^n \loss(\pred_t,y_t) - \inf_{f\in\F}\sum_{t=1}^n \loss(f(x_t),y_t)$$
	as small as possible for all sequences $y_1,\ldots,y_n$ and $x_1,\ldots,x_n$ (satisfying $x_t\in\X_t(h_{1:t-1})$).

To illustrate the setting, consider a few examples. We may take $\X_t(h_{1:t-1})=\left\{(y_1,\ldots,y_{t-1})\right\} \subset  \{0,1\}^{t-1}$ to be a singleton set containing the exact realization of the sequence so far. In this case, the choice $x_t=(y_1,\ldots,y_{t-1})$ is enforced and $f(x_t)=p_f(1|y_1,\ldots,y_{t-1})$ may be viewed as a conditional distribution; the normalized maximum likelihood forecaster is known to be minimax optimal in this extensively studied scenario (e.g. \cite[Ch. 9]{PLG}). Alternatively, we may define $\X_t(h_{1:t-1}) = \left\{y'\in \{0,1\}^{t-1}: d_{H}(y_{1:t-1},y')\leq r\right\}$ to be a set that contains histories with up to $r$ flips of the bits. In this case, the forecaster is facing a situation where history can be slightly altered in an adversarial fashion. As another example, we may take $\X_t(h_{1:t-1}) = \left\{(y_{t-k},\ldots,y_{t-1})\right\}$, in which case the forecaster competes with a set of $k$th-order stationary Markov experts. The set $\X_t$ may also be time-invariant, in which case $f$ is a memoryless expert that acts on side information. In short, the formulation we presented subsumes a wide range of interesting problems. Our goal in this paper is to understand how ``complexity'' of $\F$ affects minimax rates of regret.

The minimax regret for the problem of sequential probability assignment can be written as
	\begin{align}
		\label{eq:unrolled_minimax_with_swap}
		V_n(\F) &= \multiminimax{\sup_{x_t\in\X_t(x_{1:t-1},y_{1:t-1})}~~ \inf_{\pred_t\in[0,1]}~~ \sup_{p_t\in [0,1]}~~ \En_{y_t \sim p_t}}_{t=1}^n \left\{\sum_{t=1}^n \loss(\pred_t,y_t) - \inf_{f\in\F}\sum_{t=1}^n \loss(f(x_t),y_t)  \right\} 
	\end{align}
	where $\En_{y_t\sim p_t}$ is a shorthand for the expectation with respect to Bernoulli $y_t$ with bias $p_t$. Following \cite{StatNotes2012}, the notation $\multiminimax{\ldots}_{t=1}^n$ represents a repeated application of the operators inside the brackets and corresponds to the unrolled minimax value of the associated game between the forecaster and Nature. Any upper bound on $V_n(\F)$ guarantees existence of a strategy that attains regret of at most that amount. In the last few years, new techniques with roots in empirical process theory have emerged for analyzing minimax values of the form \eqref{eq:unrolled_minimax_with_swap}. We bring these techniques to bear on the problem of sequential probability assignment with self-information loss.

	Our point of comparison will be the study of rich classes in \cite[Section 9.10]{PLG}. Following  \cite{PLG}, we employ the truncation method to deal with the unbounded loss function. To this end, fix $\delta\in(0,1/2)$, to be chosen later. For $a\in[0,1]$, let $\thresh{a}$ denote the thresholded value
	\begin{align*}
		\thresh{a} ~=~ \begin{cases} 
		\delta &\mbox{if } a<\delta \\ 
		a &\mbox{if } a\in[\delta,1-\delta] \\
		1-\delta &\mbox{if } a>1-\delta.
	\end{cases}
	\end{align*}
	For a class $\F$, let $\F^\delta = \{\thresh{f}:f\in\F\}$ denote the class of truncated functions. It is easy to check (see \cite[Lemma 9.5]{PLG}) that 
	\begin{align}
		\label{eq:off_delta}
		V_n(\F)\leq V_n(\F^\delta)+2n\delta,
	\end{align}
	and we can, therefore, focus on the minimax regret with respect to $\F^\delta$. We show that $V_n(\F^\delta)$ can be upper bounded via a modified (offset) sequential Rademacher complexity, which in turn can be controlled via sequential chaining in the spirit of \cite{RakSri14a,RakSri15nonparametric}. Unlike the latter two papers, however, we do not employ symmetrization and instead use the self-information property of the loss function. We are able to mitigate the adverse dependence of $V_n(\F^\delta)$ on $\delta$ by introducing  chaining with Bernstein-style terms that control the sub-Gaussian and sub-exponential tail behaviors. As an example, we recover the $n^{3/5}$ rate for monotonically increasing experts presented in \cite[Sec 9.10-9.11]{PLG}. However, our technique goes well beyond such examples of ``static'' experts. In particular, we can obtain non-trivial rates even in the setting where discretization in the style of \cite[Sec 9.10-9.11]{PLG},\cite{cesa1999minimax} leads to vacuous bounds. One such example is when experts are indexed by a unit ball in a Hilbert space (or, a high-dimensional Euclidean space) and expert's prediction depends linearly on side information. A discretization in the supremum norm of this set of experts is not finite, and thus the typical approaches to this problem fail. In contrast, we employ the ideas from empirical process theory and its sequential generalization in \cite{RakSriTew14ptrf} in order to define ``data-dependent'' notions of complexity.

	Despite the improvement over the technique of \cite{PLG}, the rates attained in this paper are not always minimax optimal, as we demonstrate in Section~\ref{sec:linear}. This is in contrast to other loss functions (such as absolute, square, $q$-power, and logistic) for which matching upper and lower bounds (to within logarithmic factors) have been established recently in \cite{RakSri15nonparametric}. As mentioned in \cite{PLG}, the truncation method is crude, and we leave it as an open question whether a different technique can be employed to attain optimal rates. 

		We finish this introduction with a brief mention that sequential probability assignment is extensively studied in Information Theory, where regret is known as \emph{redundancy} with respect to a set of codes. The vast literature mostly investigates the case of parametric classes (see \cite{shtar1987universal,xie2000asymptotic,freund1996predicting,rissanen1986complexity,rissanen1996fisher} and the references in \cite[Ch. 9]{PLG}), with exact constants available in certain cases. We refer to \cite{mf-up-98} for a discussion of approaches to dealing with large comparator classes. Given the well-known connection to compression, it would be interesting to employ the relaxation-based algorithmic recipe of \cite{rakhlin2012relax,RakSri15nonparametric,StatNotes2012} to come up with novel data compression methods.

\section{Complexity of Large Classes of Experts}
	
	We focus on the minimax value for the thresholded class $\F^\delta$. To state the first technical lemma, we need the definition of a tree. For an abstract set $\Z$, a $\Z$-valued complete binary tree $\z$ of depth $n$ is a collection of labeling functions $\z_t:\{0,1\}^{t-1}\to \Z$ for $t\in\{1,\ldots,n\}$. For a sequence $y=(y_1,\ldots,y_n)\in\{0,1\}^n$ (which we call \emph{a path}), we write $\z_t(y)$ for $\z_t(y_1,\ldots,y_{t-1})$. Once we take $y_1,\ldots,y_n$ to be  random variables, we may view $\{\z_t\}$ as a predictable process with respect to the filtration given by $\sigma(y_1,\ldots,y_{t-1})$. \footnote{We remark that in \cite{RakSriTew14ptrf,RakSriTew10,StatNotes2012}, the trees are defined with respect to $\{\pm1\}$-valued sequences, whereas here we use the $\{0,1\}$-valued variables. The change is purely notational and all the definitions and results can be rephrased appropriately.}
	
	We will say that an $\x$ tree is \emph{consistent} with respect to the side information set mappings $h_{1:t-1}\mapsto \X_t(h_{1:t-1})$ if for any $y\in\{0,1\}^n$, it holds that for all $t$,
	$$\x_t(y) \in \X_t(\x_1(y),\ldots,\x_{t-1}(y), y_1,\ldots,y_{t-1}).$$
	A consistent tree respects the sets of constraints $\X_t$ imposed by the problem. For the purposes of analyzing complexity of $\F$, it is important that the constraints are reflected in the tree $\x$.

Theorem~\ref{thm:value_upper} below relates the minimax regret with respect to $\F^\delta$ to the supremum of a stochastic process of a form similar to \emph{offset Rademacher complexity} introduced in \cite{RakSri14a}. The key difference with respect to \cite{RakSri14a} is that the stochastic process is defined with potentially biased coin flips. To prove Theorem~\ref{thm:value_upper}, we avoid symmetrization and instead exploit the fact that the logarithmic loss has the self-information property: in the maximin dual, the optimal probability assignment is given precisely by the distribution of the $y_t$ variable. We note that the symmetrization approach of \cite{RakSri14a} appears to give worse rates for the logarithmic loss function. 

	Let 
	\begin{align}
		\label{def:eta}
		\eta(p,a) \deq - \ind{a=1}p^{-1}  + \ind{a=0} (1 - p)^{-1}
	\end{align} 
	and observe that $\eta$ is zero-mean if $a$ is Bernoulli random variable with bias $p$. 	
	
	\begin{theorem}
		\label{thm:value_upper}
		The following upper bound holds:
		\begin{align*}
		V_n(\F^\delta) \leq \sup_{\x, \bmu, \p}\En \sup_{f \in \F^\delta} \left[ \sum_{t: {\p_t(y)\in[\delta,1-\delta]}} \eta(\p_t(y),y_t) \left(\bmu_t(y) - f(\x_t(y))\right)  - \frac{1}{2}\left(\bmu_t(y) - f(\x_t(y))\right)^2\right] + 2 n \delta \log(1/\delta),
		\end{align*}
		where $\p,\bmu$ range over all $[0,1]$-valued trees, $\x$ ranges over consistent trees, and the stochastic process $y_1,\ldots,y_n$ is defined via $y_t|y_1,\ldots,y_{t-1} \sim \text{\sf Bernoulli}(\p_t(y_1,\ldots,y_{t-1}))$. 
	\end{theorem}

To shorten the notation in Theorem~\ref{thm:value_upper}, let $\Z = \X \times [0,1]$ and for every $f \in \F^\delta$, write $g_f(z) = g_f(x,a) = a - f(x)$. The upper bound of Theorem~\ref{thm:value_upper} can be written more succinctly as 
\begin{align}
	\label{eq:succinct_form}
	V_n(\F^\delta) \le \sup_{\z, \p}\En \sup_{f \in \F^\delta} \left[ \sum_{t: {\p_t(y)\in[\delta,1-\delta]}} \eta(\p_t(y),y_t) g_f(\z_t(y))  - \frac{1}{2}g_f(\z_t(y))^2 \right] + 2 n \delta \log(1/\delta).
\end{align}
We keep in mind that the $\x$ part of $\z$ is a consistent tree. Observe that the expression above is a supremum of a collection of random variables indexed by $f\in\F^\delta$, each with a nonpositive-mean. To analyze the supremum of this stochastic process, we first consider the case when the indexing set is finite.
\begin{lemma}
	\label{eq:finite_biased}
	For any set $V$ consisting of $[-1,1]$ valued trees, any $[\delta,1-\delta]$-valued tree $\p$, and any $c>0$,
\begin{align*}
\En_y \max_{\v \in V} \left[ \sum_{t=1}^n   \eta(\p_t(y),y_t) \v_t(y)  - c\v_t(y)^2 \right] \leq \frac{\log\ |V|}{\delta \log(1 + \frac{c}{2})} 
\end{align*}
where $y_t|y_1,\ldots,y_{t-1} \sim \text{\sf Bernoulli}(\p_t(y_1,\ldots,y_{t-1}))$. Furthermore, the same upper bound holds if $\p$ is any $[0,1]$-valued tree but the summation is restricted to $\{t: \p_t(y)\in[\delta,1-\delta]\}$.
\end{lemma}
The tight control of the expectation is possible because of the negative quadratic term that acts as a compensator. On the downside, the upper bound displays the adverse $1/\delta$ dependence. We now show a maximal inequality when the quadratic term is not present. The bound is of a Bernstein type, with the sub-Gaussian and sub-exponential behaviors. Crucially, the sub-Gaussian term scales with $1/\sqrt{\delta}$.
\begin{lemma}
	\label{eq:finite_biased_l2}
	For any set $V$ consisting of $[-1,1]$ valued trees, any $[\delta,1-\delta]$-valued tree $\p$, and any $c>0$,
\begin{align*}
\En_y \max_{\v \in V} \left[ \sum_{t=1}^n   \eta(\p_t(y),y_t) \v_t(y) \right] \leq  5\bar{v}\sqrt{\frac{n\log |V|}{\delta}} + \frac{2v_{\text{max}}\log |V|}{\delta}
\end{align*}
where $y_t|y_1,\ldots,y_{t-1} \sim \text{\sf Bernoulli}(\p_t(y_1,\ldots,y_{t-1}))$, $\bar{v} = \max_{\v\in V}\max_{y} (\frac{1}{n}\sum_{t=1}^n \v_t(y)^2)^{1/2}$, and $v_{\text{max}}=\max_{\v\in V}\max_{y} |\v_t(y)|$. The same upper bound holds if $\p$ is any $[0,1]$-valued tree but the summation is restricted to $\{t: \p_t(y)\in[\delta,1-\delta]\}$.
\end{lemma}

We now pass from a finite collection to an infinite one via the sequential chaining technique \cite{RakSriTew14ptrf}. For this purpose, we recall the definition of $\ell_p$ sequential covering numbers.
\begin{definition}[\cite{RakSriTew14ptrf}]
	\label{def:cover}
	A set $V$ of $\reals$-valued trees of depth $n$ is a (sequential) $\gamma$-cover (with respect to $\ell_p$, $p\geq 1$) of $\G\subseteq\reals^\Z$ on a $\Z$-valued tree $\z$ of depth $n$ if
	\begin{align}
		\label{eq:def_cover}
		\forall g\in\G,~~ y\in\{0,1\}^n, ~~\exists \v\in V,~~\text{s.t.}~~ \left(\frac{1}{n}\sum_{t=1}^n |\v_t(y)-g(\z_t(y))|^p \right)^{1/p}\leq \gamma.
	\end{align}
	The size of the smallest $\gamma$-cover is denoted by $\cN_p(\G,\gamma,\z)$. For $p=\infty$, \eqref{eq:def_cover} becomes $\max_{t}|\v_t(y)-g(\z_t(y))|\leq \gamma$.
\end{definition}

\begin{theorem}
	\label{thm:biased_dudley}
	Let $\G$ be a class of functions $\Z\to [-1,1]$. For any $[0,1]$-valued tree $\p$, any $\Z$-valued tree $\z$, any $K>0$, and $\gamma>0$,
	\begin{align*}
		&\En \sup_{g \in \G} \left[ \sum_{t: {\p_t(y)\in[\delta,1-\delta]}} \eta(\p_t(y),y_t) g(\z_t(y))  - K g(\z_t(y))^2 \right] \\
		&\hspace{1in}\leq \frac{1}{\delta} \frac{\log\ \cN_\infty(\G,\gamma,\z)}{ \log(1 + \frac{K}{8})}  +  \inf_{\alpha(0,\gamma]}\left\{ \frac{4 n \alpha}{\delta} + 30\sqrt{\frac{2n}{\delta}}\int_{\alpha}^\gamma \sqrt{\log\cN_\infty(\G,\rho,\z)}d\rho + \frac{8}{\delta}\int_{\alpha}^\gamma \log\cN_\infty(\G,\rho,\z)d\rho \right\} 
	\end{align*}
	where the stochastic process $y_1,\ldots,y_n$ is defined via $y_t|y_1,\ldots,y_{t-1} \sim \text{\sf Bernoulli}(\p_t(y_1,\ldots,y_{t-1}))$.
\end{theorem}

Theorem~\ref{thm:biased_dudley} is readily applied to the upper bound of Theorem~\ref{thm:value_upper} by identifying $$\G = \{g_f(z) = g_f(x,\mu) = \mu - f(x): f\in\F^\delta, \mu\in\reals, x\in\X \}$$ and $\z_t(y) = (\x_t(y),\bmu_t(y))$. It is immediate from the definition of a cover that for any $\bmu$, $\x$, and $\z=(\x,\bmu)$,
\begin{align}
	\label{eq:same_cover}
	\cN_p(\F^\delta,\x,\alpha) = \cN_p(\G,\z,\alpha).
\end{align}

The lower bound of Lemma~\ref{lem:lower_bound} (presented in Section~\ref{sec:lower}) and the relation between the offset Rademacher complexity and sequential fat-shattering dimension  \cite{RakSriTew14ptrf,RakSri15nonparametric} yield the next theorem.
\begin{theorem}
	For the case of constant sets $\X_1=\X_2=\ldots=\X$, the following are equivalent:
	\begin{itemize}
		\item Minimax regret is sublinear: $\frac{1}{n}V_n(\F)\to 0$ as $n\to \infty$
		\item Sequential dimension $\fat_\beta(\F,\X)$ is finite for all $\beta>0$
	\end{itemize}
\end{theorem}

Let us make a few remarks. First, the theorem can be easily extended to non-constant sets $\X_t$, in which case $\fat_\beta$ is defined with respect to consistent trees (as in the next section). Second, one may also phrase the equivalence through sequential covering numbers, thanks to the relations outlined in \cite{RakSriTew14ptrf,RakSri15nonparametric}. 

In summary, the sequential complexities we study are intrinsic to the problem of sequential probability assignment (unlike, for instance, covering numbers with respect to the supremum norm on $\X$ --- see Section~\ref{sec:example_history} for an example). Yet, the upper bounds we derive do not quite match the lower bounds, due to the hard thresholding approach and the need to balance $n\delta$ with $V_n(\F^\delta)$ at the end of the day. It is an open problem to close the gap between the upper and lower bounds.

The upper bound of Theorem~\ref{thm:biased_dudley} is quantified as soon as we have control of sequential covering numbers. While covering numbers could be computed directly in many situations, it is often simpler to upper bound a ``scale-sensitive dimension'' of the class, defined in the next section. In Section~\ref{sec:monotonic} we present an example of such a simple calculation.

\section{Covering Numbers and Combinatorial Parameters}
\label{sec:cov_comb}

Suppose we can define a preorder $\preceq$ on the set $\X$ (that is, a binary relation that is reflexive and transitive). We say that an $\X$-valued tree $\x$ of depth $n$ is \emph{ordered} if for any path $y\in\{0,1\}^n$, it holds that $\x_t(y)\preceq \x_{t+1}(y)$ for all $t=1,\ldots,n-1$. In this section we show that the combinatorial dimensions, covering numbers, and the associated upper bounds in \cite{RakSriTew14ptrf} can be extended to ``respect'' the preorder (of course, one can always define a vacuous relation $\preceq$ and recover prior results).

\begin{definition}
	\label{def:fat}
	A class $\F\subset \reals^\X$ shatters (at scale $\beta>0$) an ordered $\X'$-valued tree of depth $d$ if there exists a $\reals$-valued witness tree $\s$ of depth $d$ such that 
	$$\forall y\in\{0,1\}^d,~~ \exists f\in\F,~~ \mbox{s.t.}~~ (2y_t-1)(f(\x_t(y))-\s_t(y))\geq \beta/2.$$
	The largest depth of an ordered $\X'$-valued tree is denoted by $\fat_\beta^o(\F,\X')$, where the superscript $o$ stands for ``ordered''.
\end{definition}

The notion of the Littlestone's dimension $\ldim(\F,\X')$ for $\{0,\ldots,k\}$-valued function classes extends in exactly the same way to the case of ordered trees.

The main step in obtaining upper bounds on sequential covering numbers is the analogue of the Vapnik-Chervonenkis-Sauer-Shelah lemma, proved in \cite{RakSriTew14ptrf,RakSriTew10}. We now show that if we ask for a $\beta$-cover on an ordered tree $\x$, the sequential covering numbers are controlled via the ordered version $\fat_\beta^o(\F,\img(\x))$ of the fat-shattering dimension in Definition~\ref{def:fat}.

\begin{theorem}[Extension of Theorem 4 in \cite{RakSriTew14ptrf}]
	\label{thm:covering_estimate}
	Let $\F\subseteq \{0,\ldots,k\}^\X$ be a class of functions with $\fat_2^o(\F,\X)=d$. Then for any $n>d$ and any ordered $\X$-valued tree $\x$,
	$$\cN_\infty (\F,1/2,\x) \leq \sum_{i=0}^d {n\choose i}k^i.$$
	Hence, for a class $\G\subseteq[-1,1]^\X$, for any $\beta>0$,
	$$\cN_\infty(\G, \beta, \x)\leq \left(\frac{2en}{\beta}\right)^{\fat_\beta^o(\G,\X)}.$$
\end{theorem}

The following three sections are devoted to particular examples. We start by exhibiting a simple class for which sequential covering numbers are small, yet the discretization with respect to the supremum norm (typically performed to appeal to a finite-experts method) gives vacuous bounds. 

\section{Example: Consistent History}
\label{sec:example_history}

We would like to illustrate that sequential covering number can be much smaller than covering numbers with respect to the supremum norm over $\X$. Consider the particular case of $\X_t(h_{1:t-1})=\{(y_1,\ldots,y_{t-1})\}$. Clearly, there is only one consistent tree, namely the one defined by $\x_t(y) = (y_1,\ldots,y_{t-1})$ for any $t$. In this case, the requirement \eqref{eq:def_cover} in Definition~\ref{def:cover} with class $\F^\delta$, consistent tree $\x$, and $p=\infty$ reads as
	\begin{align}
		\label{eq:def_cover2}
		\forall f\in\F^\delta,~~ y\in\{0,1\}^n,~~ \exists \v\in V,~~\text{s.t.}~~ |\v_t(y_{1:t-1})-f(y_{1:t-1})| \leq \gamma.
	\end{align}
We contrast this with the definition in \cite[Sec. 9.10]{PLG}, where the covering of $\F$ is done with respect to the following pointwise metric (which we normalized by $\sqrt{n}$ for uniformity):
\begin{align}
	\label{def:sup_norm_distance}
	d(f,g) = \sqrt{\frac{1}{n}\sum_{t=1}^n \sup_{y_{1:t}} \left( \loss(f(y_{1:t-1}), y_t) - \loss(g(y_{1:t-1}), y_t)\right)^2}.
\end{align}
To illustrate a gap in the two covering-number approaches, construct a particular class $\F$ as follows. For each element $b\in\{0,1\}^n$, define $f_b$ by
$$f_b(y_{1:t-1})=\frac{1}{4}\ind{b_{1:t-1}=y_{1:t-1}}+ \frac{1}{4}$$
and take $\F=\{f_b:b\in\{0,1\}^n\}$. In other words, on round $t$, expert $f_b$ predicts probability $1/2$ if history coincides with $b_{1:t-1}$, and $1/4$ otherwise. For two elements $f_{b},f_{b'}\in\F$, let $\kappa(b,b')=\max\{t: b_t= b_t'\}$ be the last time the two sequences agree (defined as $0$ if $b_1\neq b_1'$). Then
\begin{align*}
	&\sum_{t=1}^n \sup_{y_{1:t}} \left( \loss(f_b(y_{1:t-1}), y_t) - \loss(f_{b'}(y_{1:t-1}), y_t) \right)^2 \geq \sum_{t=1}^n \left( \loss(f_b(b_{1:t-1}), 1) - \loss(f_{b'}(b_{1:t-1}), 1) \right)^2 
	\geq (n-\kappa(b,b'))\log(2)^2
\end{align*}
and thus there are at least $2^{n/2}$ functions at a constant distance $d(f,g)\geq c$. 

In contrast, consider sequential covering in the sense of \eqref{eq:def_cover2} (and Definition~\ref{def:cover}). Take any $y\in\{0,1\}^n$ and $f_b\in\F$. The sequence of $n$ values $(f_b(\emptyset), f_b(y_1),\ldots,f_b(y_{1:t-1}),\ldots, f_b(y_{1:n-1}))$ is equal to $1/2$ until $t=\kappa(b,y)$ and $1/4$ afterwards. Let $V$ be a set of $n$ trees $\v^1,\ldots,\v^n$ labeled by $\{1/4,1/2\}$. Each $\v^i$ is defined as
$$\forall y\in\{0,1\}^n, t\in\{1,\ldots,n\},~~~~~~ \v^i_t(y) = (1/4)\ind{t\leq i-1}+1/4.$$
It is immediate that this set of $n$ trees provides an exact cover of $\F$ (at scale $0$) in the sense of Definition~\ref{def:cover}. This leads to $\mathcal{O}(\log(n)/n)$ bounds on minimax regret, while the discretization with respect to the supremum norm \eqref{def:sup_norm_distance} fails. 

The above failure is endemic to approaches that attempt to discretize the set of experts before the prediction process even started. In contrast, sequential complexities can be viewed as an analogue of ``data-based'' discretization, which is known in statistical learning since the work of Vapnik and Chervonenkis in the 60's.

\section{Example: Monotonically Nondecreasing Experts}
\label{sec:monotonic}

We consider an example of a nonparametric class analyzed in \cite[p. 270]{PLG}. Let $f\in\F$ be a set of experts such that the forecasted probability does not decrease in time. To model this scenario in a general manner, we suppose that the side information $x_t = (t, x'_t)\in {\mathbb N}\times \X'_t(x_{1:t-1},y_{1::t-1})$ contains the time stamp, and $f(t+1,x'_t)\geq f(t,x''_t)$ for any $f\in\F$. The particular case of \emph{static} experts---with prediction depending only on $t$ and no other side information---has been considered in \cite{PLG}.

To invoke the results of the previous section, define a preorder on $(t,x)\in \X={\mathbb N}\times \X'$ according to the time stamp: $(t,u)\preceq (s,v)$ for any $t< s$ and $u,v\in \X'$. Suppose an ordered $\X$-valued tree $\x$ of depth $d$ is shattered, according to  Definition~\ref{def:fat}, with a witness tree $\s$. We claim that the values of the witness tree must be increasing by at least $\beta$ along the path $y=(1,1,1,\ldots)$. Indeed, consider any $t\geq 1$, and let $y'=(y_{1:t},0,y_{t+2:d})$. By the definition of shattering, there must be a function that satisfies  $f(\x_t(y'))\geq \s_t(y')+\beta/2$ and $f(\x_{t+1}(y'))\leq \s_{t+1}(y')-\beta/2$. Since $f(\x_t(y'))\leq f(\x_{t+1}(y'))$, we conclude that $\s_t(y)=\s_t(y')\leq \s_{t+1}(y')-\beta = \s_{t+1}(y)-\beta$. Hence, $\s_t$ increases by at least $\beta$ along the path $(1,\ldots,1)$ and thus $d\leq 1/\beta$. This quick calculation gives $\fat_\beta^o(\F,\X)\leq 1/\beta$.

In view of Theorem~\ref{thm:covering_estimate},
$$\log\cN_\infty(\F^\delta,\beta,\x)\leq (1/\beta)\log\left(2en/\beta\right)$$
In view of \eqref{eq:same_cover}, the same covering number estimate holds for $\G$. Then Theorem~\ref{thm:biased_dudley} with $\alpha=1/n$ and $\gamma=n^{-a}$ (with $a$ to be determined later) implies that 
$$\int_{\alpha}^\gamma \log\cN_\infty(\G,\rho,\z)d\rho \leq C\log^2 n$$
is a lower order term, with $C$ being an absolute constant. We also have
$$\int_{\alpha}^\gamma \sqrt{\log\cN_\infty(\G,\rho,\z)}d\rho\leq C'\sqrt{\log n} \cdot \gamma^{1/2}.$$
Now, ignoring constants and logarithmic terms, this gives the overall rate of 
$$\mathcal{O}^*\left(\frac{1}{\delta\gamma} + \sqrt{\frac{n\gamma}{\delta}}\right) = \mathcal{O}^*\left(n^{1/3}\delta^{-2/3}\right)$$
for the minimax regret with respect to $\F^\delta$. The terms are balanced by choosing 
$\gamma=n^{-1/3}\delta^{-1/3}$. The rate with respect to $\F$ is then 
$$\mathcal{O}^*\left(n\delta + n^{1/3}\delta^{-2/3}\right) = \mathcal{O}^*\left(n^{3/5}\right)$$
by choosing $\delta=n^{-2/5}$. This corresponds to the rate obtained by \cite{PLG}.

\section{Example: Linear Prediction}
\label{sec:linear}

In this section we consider the special case of $\X_1=\ldots=\X_n=\X=B_2$ and
\begin{align}
	\label{eq:lin_class}
	\F = \{f(x)=(\inner{w,x}+1)/2: w\in B_2\}
\end{align}
where $B_2$ is a unit Euclidean (or Hilbert) ball. Written as a function of $w$, the loss at time $t$ is (up to an additive constant $\log(2)$) 
\begin{align}
	\label{eq:funcs}
	g_t(w) = -\ind{y_t=1}\log(1+\inner{w,x_t}) - \ind{y_t=0}\log(1-\inner{w,x_t}).
\end{align}
It is possible to estimate the sequential $\fat_\beta$ dimension of a unit Hilbert ball as $\fat_\beta = \mathcal{O}^*(1/\beta^2)$, where the $\mathcal{O}^*$ notation ignores logarithmic factors. Then Theorem~\ref{thm:biased_dudley} gives an upper bound of 
$$V_n(\F^\delta) = \mathcal{O}^*\left(n^{1/2}\delta^{-1}\right),$$ 
and thus
$$V_n(\F) = \mathcal{O}^*\left(n^{3/4}\right).$$ 
Below, we exhibit an algorithm that attains regret of $\mathcal{O}^*\left(n^{1/2}\right)$, implying that the upper bounds obtained with our technique are not always tight. 

\subsection{Algorithm: Regularization with Self-Concordant Barrier}

To develop an algorithm for the problem, we turn to the field of online convex optimization. We observe that functions $g_t$ defined in \eqref{eq:funcs} are convex, but not strongly convex. Moreover, the gradients of $g_t(w)$ are not bounded. We may consider a restricted set to mitigate the exploding gradient; however, a $\delta$-shrinkage of the ball $B_2$ still leaves the gradient to be of size $O(1/\delta)$. A direct gradient descent method will give the suboptimal  $O(n^{3/4})$ upper bound derived above in a non-constructive way. We also mention that while the functions are exp-concave, the upper bounds for the Online Newton Step method  \cite{hazan2007logarithmic} scale with the dimension of the space, which we assume to be large or infinite.

We now present an algorithm based on self-concordant barrier regularization, which appears to be of an independent interest. The algorithm answers the following question: \emph{can one obtain regret bounds for online convex optimization in terms of the maximum of function values rather than gradients?}

Consider the Follow-the-Regularized-Leader method
\begin{align}
	\label{eq:ftrl}
	w_{t+1} = \argmin{w\in B_2} \sum_{s=1}^t \inner{\nabla g_s(w_s), w} + \eta^{-1} R(w)
\end{align}
with the self-concordant barrier $R(w) = -\log(1-\|w\|^2)$. In accordance with the protocol of the probability assignment problem, we predict $\inner{w_{t},x_{t}}$ at round $t$ after observing $x_t$. It is shown in \cite{AbeRak09colt} that regret of \eqref{eq:ftrl} against any $w^*\in B_2$ is 
\begin{align}
	\label{eq:reg_bound_local}
	\sum_{t=1}^n g_t(w_t) - g_t(w^*) \leq 2\eta \sum_{t=1}^n \|\nabla g_t(w_t)\|_{w_t}^{*2} + \eta^{-1} R(w^*)
\end{align}
as long as $\eta$ satisfies $\eta \|\nabla g_t(w_t)\|^*_{w_t}\leq 1/4$. Here, the \emph{local norm} is defined as 
$$\|h\|^*_{w} = \sqrt{h^\tr (\nabla^2 R(w))^{-1} h}.$$
According to the lemma below, the local norm is bounded by a constant that is independent of the dimension: 
\begin{lemma}
	\label{eq:bdd_local_norm}
	For any $t$, the local norm of $\nabla g_t(w_t)$ is upper bounded by a constant:
		$$\|\nabla g_t(w_t)\|^*_{w_t} \leq 3.$$
\end{lemma}
Together with \eqref{eq:reg_bound_local}, Lemma~\ref{eq:bdd_local_norm} implies a regret bound of 
$18\eta n + \eta^{-1} R(w^*).$
Instead of taking $w^*$ at the boundary of the ball where $R(w^*)$ is infinite, we can evaluate regret against $w=(1-1/n)w^*$. For such a comparator, $R(w)=\mathcal{O}(\log n)$. By choosing $\eta$ appropriately and using an argument similar to \eqref{eq:off_delta}, we conclude that  regret against any $w^*\in B_2$ is
upper bounded by 
$$C\sqrt{n\log n}.$$
Importantly, $C$ is an absolute constant that does not depend on the dimension of the problem. This rate is optimal up to polylogarithmic factors. The optimality follows from Lemma~\ref{lem:lower_bound} below and an estimate on sequential covering number of a Hilbert ball \cite{RakSri14a,RakSri15nonparametric}.
 \begin{lemma}
	For the linear class in \eqref{eq:lin_class}, 
$$V_n(\F) = \mathcal{\Theta}^*(n^{1/2}).$$
\end{lemma}

The proof of Lemma~\ref{eq:bdd_local_norm} relied heavily on the ability to calculate the gradient of the loss function and match it to the inverse Hessian of the self-concordant barrier. We now give an alternative proof based on a simple and charming, yet unexpected lemma due to Nesterov (see Appendix for the short proof):
\begin{lemma}[Lemma 4 in \cite{nesterov2011barrier}]
	\label{lem:nesterov}
	Let $\psi$ be concave and positive on $\text{int}~ \K$. Then for any $x\in\text{int}~ \K$ we have
	$$\|\nabla\psi(x)\|_x^* \leq \psi(x).$$
\end{lemma}

The lemma allows us to upper bound regret in an online convex optimization problem if we only know that the values of the functions (and not the gradients) are bounded. Consider the FTRL algorithm \eqref{eq:ftrl}, but over the shrunk ball $(1-1/n)B_2$. Suppose we can ensure $0< g_t < A$. Then $A-g_t$ is concave and positive. Hence, by above lemma
$$\|\nabla g_t(w_t)\|^*_{w_t} = \|\nabla (A-g_t(w_t))\|^*_{w_t} \leq A-g_t(w_t) \leq A$$
which provides an alternative to the bound of Lemma~\ref{eq:bdd_local_norm}.
Regret is then upper bounded by 
$$\sum_{t=1}^n g_t(w_t) - \sum_{t=1}^n g_t(w^*) \leq 2\eta n A^2 + \eta^{-1} R(w^*) $$
Crucially, by employing self-concordant regularization, we avoid paying for a large gradient of cost functions at the boundary of the set. Over the shrunk set $(1-1/n)B_2$, we ensure that the values of functions $g_t$ are upper bounded by $A=O(\log n)$ even if the gradients blow up linearly with $n$. This surprising observations leads to a dimension-independent $O(\sqrt{n}\log n)$ regret bound for the Euclidean ball, and can also be used for other convex bodies and non-logarithmic loss functions when the closed-form analysis of Lemma~\ref{eq:bdd_local_norm} is not available.

\section{A Lower Bound}
\label{sec:lower}

In this section, we show that the offset sequential Rademacher complexity serves as a lower bound on the minimax regret. Hence, the complexities of the class $\F$ of experts are intrinsic to the problem. We refer to \cite{RakSri15nonparametric} for further lower bounds on the offset Rademacher complexity via the scale-sensitive dimension and sequential covering numbers.

\begin{lemma}
	\label{lem:lower_bound}
	The following lower bound holds:
	\begin{align*}
		V_n(\F) +1 \geq \sup_{\x} \En_{y} \left[
				\sup_{f\in\F^{1/n}}
					\left\{
						\sum_{t=1}^n  2 (2y_t-1)(f(\x_t(y))-1/2)  - 4(\log n)(f(\x_t(y))-1/2)^2  
					\right\}
				\right]
	\end{align*}
	where $y_1,\ldots,y_n$ are independent with distribution ${\sf Bernoulli(1/2)}$ and the supremum is taken over consistent trees with respect to constraints $\X_t$.
\end{lemma}
\begin{proof}[\textbf{Proof of Lemma~\ref{lem:lower_bound}}]
To prove the lower bound, we proceed as in \cite{RakSri15nonparametric}. First, we observe that \eqref{eq:off_delta} holds in the other direction too:
	\begin{align}
		\label{eq:off_delta_opposite}
		V_n(\F^\delta)\leq V_n(\F)+n\delta.
	\end{align}
To see this, note that any $f$ only loses from thresholding when either $f(x_t)>1-\delta$ and $y_t=1$, or when $f(x_t)<\delta$ and $y_t=0$. In both cases, the difference in logarithmic loss is at most $-\log(1-\delta)\leq \delta$ for $\delta<1/2$. For the purposes of a lower bound, we take $\delta = 1/n$ and turn to lower-bounding $V_n(\F^{1/n})$.

As in the development leading to \eqref{eq:minimax_swapped} in the proof of the upper bound, the minimax value $V_n(\F^{1/n})$ is \emph{equal} to 
	\begin{align}
		&\multiminimax{\sup_{x_t}\sup_{p_t\in[0,1]} \En_{y_t}}_{t=1}^n 
			\left[
				\sup_{f\in\F^{1/n}}
					\left\{
						\sum_{t=1}^n \inf_{\pred_t\in[0,1]} \En_{y_t}\left[\loss(\pred_t,y_t)\right]  - \sum_{t=1}^n \loss(f(x_t),y_t)  
					\right\}
				\right]  
	\end{align}
which, by the self-information property of the loss equal to 
	\begin{align}
		&\multiminimax{\sup_{x_t}\sup_{p_t\in[0,1]} \En_{y_t}}_{t=1}^n 
			\left[
				\sup_{f\in\F^{1/n}}
					\left\{
						\sum_{t=1}^n \En_{y_t}\left[\loss(p_t,y_t)\right]  - \sum_{t=1}^n \loss(f(x_t),y_t)  
					\right\}
				\right]  
	\end{align}
By the linearity of expectation (and since the terms $\En_{y_t}\left[\loss(p_t,y_t)\right]$ do not involve $f$), we have
	\begin{align}
		\label{eq:lowerbd}
		V_n(\F^{1/n}) = &\multiminimax{\sup_{x_t}\sup_{p_t\in[0,1]} \En_{y_t}}_{t=1}^n 
			\left[
				\sup_{f\in\F^{1/n}}
					\left\{
						\sum_{t=1}^n \loss(p_t,y_t)  - \sum_{t=1}^n \loss(f(x_t),y_t)  
					\right\}
				\right] .
	\end{align}
	We now pass to the first lower bound by choosing $p_t=1/2$ for all $t$. 
	
	Consider the case $y_t=1$ and expand the loss function around $p_t=1/2$ for $z\in[1/n,1]$:
	\begin{align}
		\label{eq:remainder}
		\loss(1/2,1)  - \loss(z,1) = -\log(1/2) - (-\log(z)) = 2(z-1/2) - R(z)
	\end{align}
	where $R(z)$ is the remainder. We claim that the remainder can be upper bounded by a quadratic over the interval $[1/n, 1]$. To this end, consider the function
	$$g(z) = -2z + (1+\log(2))+4(\log n)(z-1/2)^2$$
	and note that the derivative and the value of this function at $1/2$ coincide with the derivative and the value of $-\log(z)$ at the same point. We claim that $g(z)$ dominates $-\log(z)$ on $[1/n,1]$. For $z>1/2$, this follows from $g'>(-\log)'$. The same argument holds for the interval $[1/\log (n), 1/2]$. Now, at $z=1/n$, $g(z)> -\log(z)$ and $|g'(z)|<|\log(z)'|$. The derivative relation continues to hold on the interval $[1/n,c/\log (n)]$ for large enough $c$, establishing $g>-\log$ on this interval too. The remaining interval $[c/\log(n), 1/\log(n)]$ is easily checked by the direct computation of function value. In sum, 
the remainder in \eqref{eq:remainder} can be upper bounded by $R(z)\leq 4(\log n)(z-1/2)^2$.

The case of $y_t=0$ is exactly analogous, and we obtain
\begin{align}
	\loss(p_t,y_t)  - \sum_{t=1}^n \loss(f(x_t),y_t)  &\geq  2\left[ \ind{y_t=1}(f(x_t)-1/2) + \ind{y_t=0}(-f(x_t)+1/2) \right] - 4(\log n)(f(x_t)-1/2)^2 \\
	&= 2\left[ y_t(f(x_t)-1/2) + (1-y_t)(-f(x_t)+1/2) \right] - 4(\log n)(f(x_t)-1/2)^2 \\
	&= 2 (2y_t-1)(f(x_t)-1/2)  - 4(\log n)(f(x_t)-1/2)^2. 
\end{align}
The lower bound in \eqref{eq:lowerbd} then becomes
	\begin{align}
		\label{eq:lowerbd}
		V_n(\F^{1/n}) &\geq \multiminimax{\sup_{x_t\in\X_t(x_{1:t-1},y_{1:t-1})} \En_{y_t}}_{t=1}^n 
			\left[
				\sup_{f\in\F^{1/n}}
					\left\{
						\sum_{t=1}^n  2 (2y_t-1)(f(x_t)-1/2)  - 4(\log n)(f(x_t)-1/2)^2  
					\right\}
				\right]\\
			&= \sup_{\x} \En_{y} \left[
				\sup_{f\in\F^{1/n}}
					\left\{
						\sum_{t=1}^n  2 (2y_t-1)(f(\x_t(y))-1/2)  - 4(\log n)(f(\x_t(y))-1/2)^2  
					\right\}
				\right]
	\end{align}
	where $y_1,\ldots,y_n$ are independent with distribution ${\sf Bernoulli(1/2)}$. 
	
\end{proof}

\section{Discussion and Open Questions}

At the very first step, the analysis in this paper thresholds the class $\F$ to avoid dealing with the exploding gradient of the loss function. The authors believe that this ``hard thresholding'' approach is the source of sub-optimality, and that ``smooth'' approaches should be possible. When the class of functions has a specific structure, such as in the example of Section~\ref{sec:linear}, the exploding gradient can be mitigated in a ``smooth way'' by a regularization technique. It is not clear to the authors how to perform the ``smooth thresholding'' analysis when such a structure is not available.

Another interesting venue of investigation is the development of algorithms. It has been shown that the minimax analysis, of the type performed in this paper, can be made constructive \cite{rakhlin2012relax,RakSri15nonparametric,StatNotes2012}. It appears that the relaxation approach may yield new (and possibly computationally efficient) methods for sequential probability assignment and data compression.

\appendix

\section{Proofs}

	\begin{proof}[\textbf{Proof of Theorem~\ref{thm:value_upper}}]
	Let us use the shorthand $\cD = [\delta,1-\delta]$. The value $V_n(\F^\delta)$ can be upper bounded by 
	\begin{align}
		\label{eq:bd1}
\multiminimax{\sup_{x_t}\inf_{\pred_t\in\cD}\sup_{p_t\in[0,1]}\En_{y_t \sim p_t}}_{t=1}^n \left\{\sum_{t=1}^n \loss(\pred_t,y_t) - \inf_{f\in\F^\delta}\sum_{t=1}^n \loss(f(x_t),y_t)  \right\} 
	\end{align}
	simply because each infimum is taken over a smaller set. Henceforth, it will be understood that $x_t$ ranges over $\X_t(x_{1:t-1},y_{1:t-1})$. The expression in \eqref{eq:bd1} is equal to
	\begin{align}
		&\multiminimax{\sup_{x_t}\sup_{p_t\in[0,1]} \En_{y_t}}_{t=1}^n 
			\left\{
				\sum_{t=1}^n \inf_{\pred_t\in\cD}  \En_{y_t}\left[\loss(\pred_t,y_t) \right] - \inf_{f\in\F^\delta}\sum_{t=1}^n \loss(f(x_t),y_t)  
			\right\} 
	\end{align}
	by an argument that can be found in \cite{AbeAgaBarRak09,RakSriTew10}. Here, it is understood that $y_t$ is a Bernoulli random variable with distribution $p_t$. Taking the infimum outside the negative sign, the above quantity is equal to
	\begin{align}
		\label{eq:minimax_swapped}
		&\multiminimax{\sup_{x_t}\sup_{p_t\in[0,1]} \En_{y_t}}_{t=1}^n 
			\left[
				\sup_{f\in\F^\delta}
					\left\{
						\sum_{t=1}^n \inf_{\pred_t\in\cD} \En_{y_t}\left[\loss(\pred_t,y_t)\right]  - \sum_{t=1}^n \loss(f(x_t),y_t)  
					\right\}
				\right]  
	\end{align}
	We now claim that each infimum in \eqref{eq:minimax_swapped} is achieved at $\pred_t=\thresh{p_t}$. Indeed, this follows because the unconstrained minimizer over $[0,1]$ is $p_t$ by the well-known property of entropy:
	$$\argmin{\pred_t\in [0,1]} \En_{y_t}\left[\loss(\pred_t,y_t)\right] = \argmin{\pred_t\in [0,1]} \Big\{ -p_t\log(\pred_t) - (1-p_t)\log(1-\pred_t) \Big\}= p_t.$$
	We conclude that \eqref{eq:minimax_swapped} is equal to 
	\begin{align}
		\label{eq:minimax_swapped2}
		&\multiminimax{\sup_{x_t}\sup_{p_t\in[0,1]} \En_{y_t}}_{t=1}^n 
			\left[
				\sup_{f\in\F^\delta}
					\left\{
						\sum_{t=1}^n \En_{y_t}\left[\loss(\thresh{p_t},y_t)\right]  - \sum_{t=1}^n \loss(f(x_t),y_t)  
					\right\}
				\right].
	\end{align}
	Now, the terms in the first sum do not depend on $f\in\F$, and thus can pass through the multiple infima and suprema. By linearity of expectation, \eqref{eq:minimax_swapped2} is equal to
	\begin{align}
		\label{eq:minimax_swapped3}
		&\multiminimax{\sup_{x_t}\sup_{p_t\in[0,1]} \En_{y_t}}_{t=1}^n 
			\left[
				\sup_{f\in\F^\delta}
					\left\{
						\sum_{t=1}^n \loss(\thresh{p_t},y_t) - \sum_{t=1}^n \loss(f(x_t),y_t)  
					\right\}
				\right]
	\end{align}
	We now separately deal with the case that $p_t\notin\cD$. To this end, observe that
	\begin{align*}
		\ind{p_t<\delta}(\loss(\thresh{p_t},y_t) - \loss(f(x_t),y_t)) &= \ind{p_t<\delta,y_t=0}(\loss(\delta,y_t) - \loss(f(x_t),y_t)) \\
		&+ \ind{p_t<\delta,y_t=1}(\loss(\delta,y_t) - \loss(f(x_t),y_t))\\
		&\leq \ind{p_t<\delta,y_t=1}(\loss(\delta,1) - \loss(f(x_t),1))\\
		&\leq -\ind{p_t<\delta,y_t=1}\log\delta
	\end{align*}
	The first inequality is obtained by dropping the non-positive term. Indeed,  $p_t<\delta$ gives higher odds to the outcome $y_t=0$ than $f(x_t)\geq \delta$. Positivity of $\loss$ gives the second inequality. A similar calculation gives
	$$\ind{p_t>1-\delta}(\loss([p_t],y_t) - \loss(f(x_t),y_t)) \leq -\ind{p_t>1-\delta, y_t=0}\log\delta$$
	Substituting into \eqref{eq:minimax_swapped3}, we obtain an upper bound of
	\begin{align}
		\multiminimax{\sup_{x_t}\sup_{p_t\in[0,1]} \En_{y_t}}_{t=1}^n 
		\left[
			\sup_{f\in\F^\delta}
				\left\{
					\sum_{t=1}^n \ind{p_t\in\cD}(\loss(\thresh{p_t},y_t) - \loss(f(x_t),y_t))   -\ind{p_t<\delta,y_t=1}\log\delta -\ind{p_t>1-\delta,y_t=0}\log\delta
				\right\}
			\right] \notag		
	\end{align}
	Since $$\En_{y_t\sim p_t} \ind{p_t<\delta,y_t=1}\log(1/\delta) \leq \delta\log(1/\delta),$$
	and since $\ind{p_t\in\cD}(\loss(\thresh{p_t},y_t)=\ind{p_t\in\cD}(\loss(p_t,y_t)$, we conclude that the minimax value $V_n(\F^\delta)$ is upper bounded by
	\begin{align}
		\label{eq:minimax_swapped4}
		&\multiminimax{\sup_{x_t}\sup_{p_t\in[0,1]} \En_{y_t}}_{t=1}^n 
		\left[
			\sup_{f\in\F^\delta}
				\left\{
					\sum_{t=1}^n \ind{p_t\in\cD}(\loss(p_t,y_t) - \loss(f(x_t),y_t))   
				\right\}
			\right] + 2n\delta\log(1/\delta).		
	\end{align}
	We now linearize the terms $\loss(p_t,y_t) - \loss(f(x_t),y_t)$. The derivative of $\loss(\cdot,y_t)$ at $p_t$ is
	$$\loss'(p_t,y_t) = -\ind{y_t=1}\frac{1}{p_t} + \ind{y_t=0}\frac{1}{1-p_t}.$$
	Observe that the second derivative $\loss''(\cdot,y_t) \geq 1$, and hence $\loss(\cdot,y_t)$ is strongly convex, for either value of $y_t$. Strong convexity implies that
	$$\loss(p_t,y_t) - \loss(f(x_t),y_t)\leq \loss'(p_t,y_t) \cdot (p_t-f(x_t)) - \frac{1}{2}(p_t-f(x_t))^2$$	
	and thus \eqref{eq:minimax_swapped4} is upper bounded by 
	\begin{align}
		\label{eq:uppbd1}
		 &\multiminimax{\sup_{x_t}\sup_{p_t\in[0,1]} \En_{y_t}}_{t=1}^n\left[\sup_{f\in\F^\delta} \sum_{t: {p_t\in\cD}}  \loss'(p_t,y_t) \cdot \left( p_t - f(x_t) \right) - \frac{1}{2}(p_t - f(x_t))^2 \right]  + 2n\delta\log(1/\delta).
	 \end{align}
	Observe that the derivatives are mean-zero:
	\begin{align}
		\En_{y_t\sim p_t} \loss'(p_t,y_t) = \En\left[-\ind{y_t=1}\frac{1}{p_t} + \ind{y_t=0}\frac{1}{1-p_t}\right] = 0,
	\end{align}
	which suggests that we can symmetrize these terms as in \cite{RakSri14a,RakSri15nonparametric}. The key observation is that tighter control on the supremum over $\F$ will be obtained if we keep the derivatives to have a non-uniform distribution given by $p_t$.
	
	Let us drop the term $2n\delta\log(1/\delta)$ in \eqref{eq:uppbd1} and concentrate on the first term. Consider the following upper bound:
	\begin{align}
		\label{eq:uppbd2}
		 &\multiminimax{\sup_{x_t}\sup_{p_t\in[0,1]} \En_{y_t}}_{t=1}^n\left[\sup_{f\in\F^\delta} \sum_{t: {p_t\in\cD}}  \loss'(p_t,y_t) \cdot \left( p_t - f(x_t) \right) - \frac{1}{2}(p_t - f(x_t))^2 \right]  \notag\\
		 &\leq \multiminimax{\sup_{x_t}\sup_{p_t,p_t'\in[0,1]} \En_{y_t\sim p_t'}}_{t=1}^n\left[\sup_{f\in\F^\delta} \sum_{t: {p_t\in\cD}}  \loss'(p_t,y_t) \cdot \left( p_t' - f(x_t) \right) - \frac{1}{2}(p_t - f(x_t))^2 \right]  
	 \end{align}
	This upper bound holds because the supremum allows the choice $p_t=p'_t$ in addition to distinct choices for the two distributions.
	
	We now pass to the tree notation. Observe that the optimal choice of $x_t,p_t,p'_t$ depends on $(y_1,\ldots,y_{t-1})\in\{0,1\}^{t-1}$. In the functional form, let $\x$ be a sequence of mappings $\x_1,\ldots,\x_n$ with the consistency property $\x_t(y_1,\ldots,y_{t-1})\in \X_t(\x_1(y),\ldots,\x_{t-1}(y),y_{1:t-1})$ for all $y_{1:t-1}$. Similarly, let $\bmu$ and $\p$ be sequences of mappings with $\bmu_t,\p_t:\{0,1\}^{t-1}\to[0,1]$. With the same reasoning as in \cite{RakSriTew10}, we can write \eqref{eq:uppbd2} as
\begin{align*}
&\sup_{\x, \bmu, \p}\En \sup_{f \in \F^\delta} \left[ \sum_{t: {\p_t(y)\in\cD}} \loss'(\p_t(y),y_t) \left(\bmu_t(y) - f(\x_t(y))\right)  - \frac{1}{2}\left(\bmu_t(y) - f(\x_t(y))\right)^2\right]
\end{align*}
where $y_t$'s in $\{0,1\}$ are drawn from $\p$. More specifically, $y_1 \sim \p_1$ and subsequently $y_t \sim \p_t(y_{1:t-1})$. 
\end{proof}

\begin{proof}[\textbf{Proof of Lemma~\ref{eq:finite_biased}}]
\begin{align}
	\label{eq:softmax}
\En \sup_{\v \in V} \left[ \sum_{t=1}^n   \eta(\p_t(y),y_t) \v_t(y)  - c\v_t(y)^2 \right] &= \En \inf_{\lambda >0}\frac{1}{\lambda}\log\left( \sum_{\v \in V}\exp\left( \lambda \sum_{t=1}^n   \eta(\p_t(y),y_t) \v_t(y)  - c\v_t(y)^2\right) \right) \notag\\
&\le \inf_{\lambda > 0} \frac{1}{\lambda} \log\left( \sum_{\v \in V}\En \prod_{t=1}^n  \exp\left( \lambda \left(\eta(\p_t(y),y_t) \v_t(y)  - c\v_t(y)^2\right)\right) \right).
\end{align}
Let $X$ be a zero-mean random variable taking on a value $-v/p$ with probability $p$ and $v/(1-p)$ with probability $(1-p)$, where $\delta< p < 1/2$ and $|v|\leq 1$.
From the fact that $(e^x - x -1)/x^2$ is a non-decreasing function and $|X|<1/\delta$ almost surely, it follows that
$$
e^{\lambda X} - \lambda X -1 \le  \delta^2 X^2\ \left(e^{\lambda/\delta } - \lambda/\delta -1\right).
$$
Taking expectation over $X$ and upper bounding the variance $\En X^2 \leq 2p(v/p)^2 \leq 2v^2/\delta$,
$$
\En e^{\lambda X} -1 \leq 2 v^2 \delta \left(e^{\lambda/\delta } - \lambda/\delta -1\right).
$$
Using $1+x \le e^x$,
$$
\En e^{\lambda X} \le  \exp\left\{ 2v^2 \delta \left(e^{\lambda/\delta } - \lambda/\delta -1\right) \right\}.
$$
Applying the above derivation,
\begin{align*}
	\En\left[ \exp\left( \lambda \left(\eta(\p_t(y),y_t) \v_t(y)  - c\v_t(y)^2\right)\right) ~\middle|~ y_1,\ldots,y_{t-1}\right] 
	&= \exp\left(- \lambda c \v_t(y)^2\right) \times \En\left[ \exp\left( \lambda \eta(\p_t(y),y_t) \v_t(y) \right) ~\middle|~ y_1,\ldots,y_{t-1}\right] \\
	&\leq \exp\left(- \lambda c \v_t(y)^2\right) \times \exp\left\{ 2\v_t(y)^2 \delta \left(e^{\lambda/\delta } - \lambda/\delta -1\right) \right\}\\
  & = \exp\left(2 \delta\ \v_t(y)^2 \left(  e^{\frac{\lambda}{\delta} } - 1 - \left(1 + \frac{c}{2}\right)\frac{ \lambda}{\delta} \right)\right).
\end{align*}
Choosing $\lambda = \log(1 + \frac{c}{2})  \delta$ we ensure that $\left(  e^{\frac{\lambda}{\delta} } - 1 - \frac{(2+c) \lambda}{2 \delta} \right) <0$ and 

$$\En\left[ \exp\left( \lambda \left(\eta(\p_t(y),y_t) \v_t(y)  - c\v_t(y)^2\right)\right) ~\middle|~ y_1,\ldots,y_{t-1}\right] \leq 1.
$$ 
Iterating the argument from $t=n$ down to $t=1$ in \eqref{eq:softmax}, we obtain
\begin{align*}
\En \sup_{\v \in V} \left[ \sum_{t=1}^n   \eta(\p_t(y),y_t) \v_t(y)  - c\v_t(y)^2\right] \le \frac{\log\ |V|}{\delta \log(1 + \frac{c}{2})}. 
\end{align*}
The case when $\p$ is $[0,1]$-valued, but the summation is taken only over $\{t: \p_t(y)\in[\delta,1-\delta]\}$, follows immediately through the same argument.
\end{proof}

\begin{proof}[\textbf{Proof of Lemma~\ref{eq:finite_biased_l2}}]
Both sides of the inequality in the statement of the Lemma are homogenous with respect to $v_{\text{max}}$, and so we can assume $v_{\text{max}}=1$ and rescale the problem. We have
\begin{align}
	\label{eq:moment_gen}
\En_y \max_{\v \in V} \left[ \sum_{t=1}^n   \eta(\p_t(y),y_t) \v_t(y) \right] &\leq \inf_{\lambda>0} \left\{\frac{1}{\lambda}\log \sum_{\v\in V} \En \exp\left(\lambda \sum_{t=1}^n   \eta(\p_t(y),y_t) \v_t(y) \right) \right\}  \notag\\
&\leq \inf_{\lambda>0} \left\{\frac{\log |V|}{\lambda} + \max_{\v\in V} \frac{1}{\lambda}\log  \En \exp\left(\lambda \sum_{t=1}^n   \eta(\p_t(y),y_t) \v_t(y) \right) \right\}.
\end{align}
As shown in the proof of Lemma~\ref{eq:finite_biased}, if $X$ is a zero-mean random variable taking on a value $-v/p$ with probability $p$ and $v/(1-p)$ with probability $(1-p)$, where $\delta< p < 1/2$ and $|v|\leq 1$, then
$$
\log \En e^{\lambda X} \le 2v^2 \delta \phi(\lambda/\delta)
$$
where $\phi(x)= e^x-x-1$. Hence,
\begin{align*}
	 \En \left[ \exp\left(\lambda \sum_{t=1}^n   \eta(\p_t(y),y_t) \v_t(y) \right) ~\middle| y_1,\ldots,y_{n-1}\right] &\leq \exp\left(\lambda \sum_{t=1}^{n-1} \eta(\p_t(y),y_t) \v_t(y) \right) \times \En \left[ \exp\left(\lambda \eta(\p_n(y),y_n) \v_n(y) \right) ~\middle| y_1,\ldots,y_{n-1}\right]\\
	 &\leq \exp\left(\lambda \sum_{t=1}^{n-1} \eta(\p_t(y),y_t) \v_t(y) \right) \times \exp\left\{ 2 \delta \phi(\lambda/\delta) \max_{y_{n-1}} \v_n(y)^2 \right\}
\end{align*}
For $y_{n-1}$, we proceed in a similar fashion:
\begin{align*}
	 &\En \left[ \exp\left(\lambda \sum_{t=1}^{n-1}  \eta(\p_t(y),y_t) \v_t(y) \right) \times \exp\left\{ 2 \delta \phi(\lambda/\delta) \max_{y_{n-1}} \v_n(y)^2 \right\}  ~\middle| y_1,\ldots,y_{n-2}\right] \\
	 &\leq \exp\left(\lambda \sum_{t=1}^{n-2}  \eta(\p_t(y),y_t) \v_t(y) \right) \times \En \left[ \exp\left\{ \lambda \eta(\p_{n-1}(y),y_{n-1}) \v_{n-1}(y) +  2 \delta \phi(\lambda/\delta) \max_{y_{n-1}} \v_n(y)^2 \right\}  ~\middle| y_1,\ldots,y_{n-2}\right] \\
	 &\leq \exp\left(\lambda \sum_{t=1}^{n-2}  \eta(\p_t(y),y_t) \v_t(y) \right) \times \exp\left\{ 2 \delta \phi(\lambda/\delta) \v_{n-1}(y)^2 +  2 \delta \phi(\lambda/\delta) \max_{y_{n-1}} \v_n(y)^2 \right\} \\
	 &\leq \exp\left(\lambda \sum_{t=1}^{n-2}  \eta(\p_t(y),y_t) \v_t(y) \right) \times \exp\left\{ 2 \delta \phi(\lambda/\delta)  \max_{y_{n-2}, y_{n-1}} \left\{  \v_{n-1}(y)^2 + \v_n(y)^2 \right\} \right\} 
\end{align*}
Unrolling the expression to $t=1$ we obtain 
$$\log \En \left[ \exp\left(\lambda \sum_{t=1}^n   \eta(\p_t(y),y_t) \v_t(y) \right) \right]\leq 2 \delta \phi(\lambda/\delta) n\bar{v}^2  $$
where $\bar{v}^2=\max_{\v\in V}\max_{y} \frac{1}{n}\sum_{t=1}^n \v_t(y)^2$. In view of \eqref{eq:moment_gen}, we get
\begin{align}
	\label{eq:choice_lambda}
\En_y \max_{\v \in V} \left[ \sum_{t=1}^n   \eta(\p_t(y),y_t) \v_t(y) \right] 
&\leq \inf_{\lambda>0} \left\{\frac{\log |V|}{\lambda} + \frac{2 \delta \phi(\lambda/\delta) n\bar{v}^2}{\lambda} \right\}.
\end{align}
First, consider the case $\delta\geq \frac{\log |V|}{4n\bar{v}^2}$. Then the choice
$\lambda = \frac{1}{2}\sqrt{\frac{\delta \log |V|}{n\bar{v}^2}}$
ensures $\lambda \leq \delta$. In this case, $\phi(\lambda/\delta)$ can be upper bounded by a quadratic $\phi(\lambda/\delta)\leq (\lambda/\delta)^2 \cdot e$. The upper bound in \eqref{eq:choice_lambda} becomes 
$$\frac{\log |V|}{\lambda} + \frac{2e (\lambda/\delta)^2 \delta  n\bar{v}^2}{\lambda} \leq (2+e)\sqrt{\frac{n\bar{v}^2\log |V|}{\delta}}.$$
On the other hand, if $\delta < \frac{\log |V|}{4n\bar{v}^2}$, the upper bound in \eqref{eq:choice_lambda} becomes
$$ \inf_{\lambda>0} \left\{\frac{\log |V|}{\lambda} + \frac{2 \log|V| \phi(\lambda/\delta) }{4\lambda} \right\}.$$
Choosing $\lambda=\delta$ yields an upper bound of
$\frac{2\log |V|}{\delta}.$ Combining the two cases, we arrive at the statement of the Lemma. 
\end{proof}

\begin{proof}[\textbf{Proof of Theorem~\ref{thm:biased_dudley}}]
	Let us use the shorthand $\cD = [\delta,1-\delta]$.  Let $V'$ be a sequential $\gamma$-cover of $\G$ on $\z$ in the $\ell_\infty$ sense, i.e.
	$$\forall y\in\{0,1\}^n,~~ \forall g\in\G,~~ \exists \v\in V' \mbox{~~~s.t.~~~} |g(\z_t(y))-\v_t(y)| \leq \gamma.$$
	Of course, an $\ell_\infty$ cover is also an $\ell_2$ cover at the same scale.
	Let us augment $V'$ to include the all-zero tree, and denote the resulting set by $V=V'\cup \{\boldsymbol{0}\}$. Denote by $\v[\epsilon,g]$ a $\gamma$-close tree promised above. We have
	\begin{align}
		\label{eq:dudley_initial_split}
		&\En \sup_{g\in\G}\left[ \sum_{t: {\p_t(y)\in\cD}} \eta(\p_t(y),y_t) g(\z_t(y))  - K g(\z_t(y))^2 \right] \\
		&=\En \sup_{g\in\G}\left[ \sum_{t: {\p_t(y)\in\cD}}  \eta(\p_t(y),y_t) \Big(g(\z_t(y))-\v[y,g]_t(y)\Big) - K\Big(g(\z_t(y))^2-\frac{1}{4}\v[y,g]_t(y)^2\Big) \right.\\
		&\left.\hspace{1in} +\Big(\eta(\p_t(y),y_t)\v[y,g]_t(y) - \frac{K}{4}\v[y,g]_t(y)^2  \Big)\right] \\
		&\leq \En \sup_{g\in\G}\left[ \sum_{t: {\p_t(y)\in\cD}}  \eta(\p_t(y),y_t) \Big(g(\z_t(y))-\v[y,g]_t(y)\Big) - K\Big(g(\z_t(y))^2-\frac{1}{4}\v[y,g]_t(y)^2\Big)\right]\\
		&+ \En \max_{\v\in V'}\left[\sum_{t: {\p_t(y)\in\cD}}  \eta(\p_t(y),y_t)\v_t(y) - \frac{K}{4}\v_t(y)^2 \right]
	\end{align}
	We now claim that for any $y$ and $g$ there exists an element $\v[y,g]\in V$ such that
	\begin{align}
		\label{eq:desired_l2_norm_relation}
		\sum_{t=1}^{n} g(\z_t(y))^2\geq \frac{1}{4}\sum_{t=1}^{n}\v[y,g]_t(y)^2
	\end{align}
	and so we can drop the corresponding negative term in the supremum over $\G$. First consider the easy case $\frac{1}{n}\sum_{t=1}^{n} g(\z_t(\epsilon))^2 \leq \gamma^2$. Then we may choose $\boldsymbol{0}\in V$ as a tree that provides a sequential $\gamma$-cover in the $\ell_2$ sense. Clearly, \eqref{eq:desired_l2_norm_relation} is then satisfied with this choice of $\v[\epsilon,g]=\boldsymbol{0}$. Now, assume $\frac{1}{n}\sum_{t=1}^{n} g(\z_t(\epsilon))^2 > \gamma^2$. Fix any tree $\v[\epsilon,g]\in V$ that is $\gamma$-close in the $\ell_2$ sense to $g$ on the path $\epsilon$. Denote $u=(\v[\epsilon,g]_1(\epsilon),\ldots, \v[\epsilon,g]_n(\epsilon))$ and $h=(g(\z_1(\epsilon)),\ldots, g(\z_n(\epsilon)))$. Thus, we have that $\|u-h\|\leq \gamma$ and $\|h\|\geq \gamma$ for the norm $\|h\|^2=\frac{1}{n}\sum_{t=1}^n h_t^2$.
	Then $$\|u\|\leq \|u-h\|+\|h\| \leq \gamma + \|h\| \leq 2 \|h\|$$
	and thus $\|h\|\geq \frac{1}{2}\|u\|$ as desired. We conclude that
	\begin{align}
		\label{eq:dudleysplit}
		&\En \sup_{g\in\G}\left[ \sum_{t: {\p_t(y)\in\cD}} \eta(\p_t(y),y_t) g(\z_t(y))  - K g(\z_t(y))^2 \right] \\
		&\leq \En \sup_{g\in\G}\left[ \sum_{t: {\p_t(y)\in\cD}}  \eta(\p_t(y),y_t) \Big(g(\z_t(y))-\v[y,g]_t(y)\Big) \right] + \En \max_{\v\in V'}\left[\sum_{t: {\p_t(y)\in\cD}}  \eta(\p_t(y),y_t)\v_t(y) - \frac{K}{4}\v_t(y)^2 \right]\notag
	\end{align}
	By Lemma~\ref{eq:finite_biased}, the second term is upper bounded by
	$$\frac{\log\ \cN_\infty(\G,\gamma,\z)}{\delta \log(1 + \frac{K}{8})} $$
	As for the second term,	we note that conditionally on $y_1,\ldots,y_{t-1}$, the random variable $\eta(\p_t(y),y_t)$ is zero-mean. Let us proceed with the chaining technique. To this end, let $\v[g,y]^j\in V^j$ be an element of a $\gamma 2^{-j}$-cover of $g\in\G$.
	\begin{align}
		&\En \sup_{g\in\G}\left[ \sum_{t: {\p_t(y)\in\cD}}  \eta(\p_t(y),y_t) \Big(g(\z_t(y))-\v[y,g]_t(y)\Big) \right] \\
		&\leq \sum_{j=1}^N \En \sup_{g\in\G}\left[ \sum_{t: {\p_t(y)\in\cD}}   \eta(\p_t(y),y_t) \Big(\v[y,g]^{j}_t(y)-\v[y,g]^{j-1}_t(y)\Big) \right] \\
		&+ \En \sup_{g\in\G}\left[ \sum_{t: {\p_t(y)\in\cD}}   \eta(\p_t(y),y_t) \Big(g(\z_t(y))-\v[y,g]^{N}_t(y)\Big) \right] 
	\end{align}
	For the last term we use the Cauchy-Schwartz inequality: for any $y$ and $g\in \G$,
	\begin{align}
		\sum_{t: {\p_t(y)\in\cD}}   \eta(\p_t(y),y_t) \Big(g(\z_t(y))-\v[y,g]^{N}_t(y)\Big) &\leq \left( \sum_{t: {\p_t(y)\in\cD}}   \eta(\p_t(y),y_t)^2 \right)^{1/2} \left(\sum_{t: {\p_t(y)\in\cD}}  \Big(g(\z_t(y))-\v[y,g]^{N}_t(y)\Big)^2\right)^{1/2}\\
		&\leq \frac{1}{\delta} n \gamma 2^{-N}
	\end{align}
	Further, for any $j=1,\ldots,N$,
	\begin{align*}
		\En \sup_{g\in\G}\left[ \sum_{t: {\p_t(y)\in\cD}}   \eta(\p_t(y),y_t) \Big(\v[y,g]^{j}_t(y)-\v[y,g]^{j-1}_t(y)\Big) \right] &\leq \En \max_{\w\in W^j} \left[ \sum_{t: {\p_t(y)\in\cD}} \eta(\p_t(y),y_t) \w_t(y) \right]
	\end{align*}
	where $W^j$ is defined as the set of difference trees, defined as follows. For each pair $\v'\in V^j, \v''\in V^{j-1}$, let $\w$ be defined for each path $(y_1,\ldots,y_n)\in\{0,1\}^n$ and $t\in\{1,\ldots,n\}$ as
	$$\w_t(y) = \begin{cases} \v'_t(y)-\v''_t(y), ~~~~~\text{if exists } (y'_t,\ldots,y'_n) ~~\text{s.t. }~~ \exists g\in\G ~~~\text{s.t.}~~~ \v'=\v[g,\bar{y}]^j, \v''=\v[g,\bar{y}]^{j-1}, \bar{y} = (y_1,\ldots,y_{t-1},y'_t,\ldots,y'_n)\\  0 ~~~~\text{otherwise}\end{cases}.$$
	In other words, $\w$ is defined for each element of the tree as the difference between two trees if there is continuation of the path on which the two trees are indeed covering elements for some $g\in\G$, and $0$ if no such continuation exists. Then $W^j$ is defined as the collection of all such trees $\w$ obtained by pairing up all choices of trees from $V^j$ and $V^{j-1}$. Clearly, the size $|W^j|\leq |V^j|\times |V^{j-1}| \leq |V^j|^2$.
		
	We now use the result of Lemma~\ref{eq:finite_biased_l2}:
	\begin{align}
		\En \max_{\w\in W^j} \left[ \sum_{t: {\p_t(y)\in\cD}} \eta(\p_t(y),y_t) \w_t(y) \right] \leq 5\bar{v}\sqrt{\frac{\log |W^j|}{\delta}}+ \frac{2v_{\text{max}}\log|W^j|}{\delta}.
	\end{align}
	with $\bar{v} = \max_{\w,y} (\sum_{t=1}^n \w_t(y)^2)^{1/2}$ and $v_{\text{max}}=\max_{\w,y}|\w_t(y)|$. We over-bound $\bar{v}$ by $v_{\text{max}}$ in the arguments below. By construction of each $\w\in W^j$, the $\ell_2$ norm along any path is upper bounded by $3\sqrt{n}\gamma 2^{-j}$ (see \cite{RakSriTew14ptrf}). We conclude that	 
	 $$\En \max_{\w\in W^j} \left[ \sum_{t: {\p_t(y)\in\cD}} \eta(\p_t(y),y_t) \w_t(y) \right]\leq 15\sqrt{\frac{2n}{\delta}} (\gamma 2^{-j})\sqrt{\log|V^j|} + \frac{4(\gamma 2^{-j})\log|V^j|}{\delta}.$$
	 Observe that
	 \begin{align}
		 \sum_{j=1}^N \gamma 2^{-j}\sqrt{\log|V^j|} &= 2 \sum_{j=1}^N (\gamma 2^{-j}-\gamma 2^{-(j+1)}) \sqrt{\log\cN_\infty(\G, \gamma 2^{-j}, \z)} \\
		 &\leq 2 \int_{\gamma 2^{-(N+1)}}^{\gamma} \sqrt{\log\cN_\infty(\G, \rho, \z)}d\rho
	\end{align}
	and similarly
	 \begin{align}
		 \sum_{j=1}^N \gamma 2^{-j}\log|V^j| &\leq 2 \int_{\gamma 2^{-(N+1)}}^{\gamma} \log\cN_\infty(\G, \rho, \z)d\rho.
	\end{align}	
	Fix $\alpha \in (0,\gamma)$ and let $N=\max\{j: \gamma 2^{-j} > 2\alpha\}$. Then 
	$\gamma 2^{-(N+1)} \leq 2\alpha$ and $\gamma 2^{-N}\leq 4\alpha$. 
    Combining all the bounds,
 	\begin{align}
 		&\En \sup_{g\in\G}\left[ \sum_{t: {\p_t(y)\in\cD}}  \eta(\p_t(y),y_t) \Big(g(\z_t(y))-\v[y,g]_t(y)\Big) \right] \\
		&\leq \inf_{\alpha(0,\gamma]}\left\{ \frac{4 n \alpha}{\delta} + 30\sqrt{\frac{2n}{\delta}}\int_{\alpha}^\gamma \sqrt{\log\cN_\infty(\G,\rho,\z)}d\rho + \frac{8}{\delta}\int_{\alpha}^\gamma \log\cN_\infty(\G,\rho,\z)d\rho \right\}
	 \end{align}
	 The statement of the theorem follows by combining the two upper bounds for \eqref{eq:dudleysplit}.
\end{proof}

\begin{proof}[\textbf{Proof of Theorem~\ref{thm:covering_estimate}}]
	The proof closely follows the one in \cite[Thm. 4]{RakSriTew14ptrf}, and we refer to that paper for the missing details. Define the function $g_k(d,n) = \sum_{i=0}^d {n\choose i}k^i$ for $n\geq 1$ and $d\geq 0$, and note the recursion
	$$g_k(d,n-1)+kg_k(d-1,n-1) = g_k(d,n).$$
	We proceed by induction on $(n,d)$. The base of the induction is the same as in the proof of \cite[Thm. 4]{RakSriTew14ptrf}. For the induction step, fix an \emph{ordered} $\X$-valued tree $\x$ of depth $n$ and suppose $\fat_2^o(\F,\X)=d$. Define the partition $\F=\cup_{i=1}^k \F_i$ according to $\F_i=\{f: f(\x_1)=i\}$. For the sake of contradiction, suppose $\fat_2^o(\F_i,\img(\x)) = \fat_2^o(\F_j,\img(\x))=d$ for some $j-i\geq 2$. Then there exists two $\img(\x)$-valued ordered trees $\w$ and $\v$ of depth $d$ that are 2-shattered by $\F_i$ and $\F_j$, respectively. Crucially, $\x_1$ cannot appear in either of these trees (that is, $\x_1\notin\img(\w)\cup\img(\v)$) because functions in $\F_i$ (resp., $\F_j$) are constant on $\x_1$. Furthermore, $\x_1\preceq a$ for any $a\in\img(\w)\cup\img(\v)$. Hence, by joining $\w$ and $\v$ with $\x_1$ at the root, we obtain an ordered tree which is now $2$-shattered. The witness of this shattering is constructed by joining the two witnesses (for $\w$ and $\v$) and $(i+j)/2$ at the root. This leads to a contradiction. The rest of the proof follows exactly as in \cite[Thm. 4]{RakSriTew14ptrf}.
\end{proof}

\begin{proof}[\textbf{Proof of Lemma~\ref{eq:bdd_local_norm}}]
The gradient of $g_t$ at $w_t$ is
$$\nabla g_t(w_t) =  -\ind{y_t=1} \frac{x_t}{1+\inner{w_t,x_t}} + \ind{y_t=0} \frac{x_t}{1-\inner{w_t,x_t}} $$
and the Hessian of the barrier as
$$\nabla^2 R(w_t) = \frac{2}{1-\|w_t\|^2}I + \frac{4}{(1-\|w_t\|^2)^2 }w_t w_t^\tr.$$
By rotational invariance, for the following calculation we may assume without loss of generality that $w_t=a\e_1$ is in the direction of the basis vector $\e_1$ and $a>0$. We can then write the inverse (see \cite{AbeRak09colt}) as
$$\nabla^2 R(w_t)^{-1} = \frac{1}{2}(1-a^2)(I-\e_1\e_1^\tr) + \frac{(1-a^2)^2}{2(1-a^2)+ 4}\e_1\e_1^\tr\preceq (1-a)(I-\e_1\e_1^\tr) + \frac{2}{3}(1-a)^2\e_1\e_1^\tr.$$
Consider the case $y_t=0$ (the analysis for $y_t=1$ follows the same lines). Let us write $x_t = b\e_1 + \y$ with $\inner{\y,\e_1}=0$ and $\|y\|^2\leq 1-b^2$. We have
$$\nabla g_t(w_t) = \frac{x_t}{1-\inner{w_t,x_t}} = \frac{b\e_1 + \y}{1-ab}.$$
and
$$\nabla g_t(w_t)^\tr \nabla^2 R(w_t)^{-1} \nabla g_t(w_t) \leq \frac{b^2}{(1-ab)^2}\cdot (1-a)^2 + \frac{1-b^2}{(1-ab)^2}\cdot (1-a)$$
If $b\leq 0$, the above expression is upper bounded by $2$, and for $b>0$, the expression is upper bounded by $3$ (we did not optimize the constants). 
\end{proof}

\begin{proof}[\textbf{Proof of Lemma~\ref{lem:nesterov}}]
	We reproduce the proof from \cite{nesterov2011barrier} for completeness. Let $x\in\text{int}~ \K$ and $r\in[0,1)$. Let
	$$y = x - \frac{r}{\|\nabla\psi(x)\|_x^*} [\nabla^2 F(x)]^{-1}\nabla \psi(x).$$
	Then $y\in\text{int} \K$ because the Dikin ellipsoid is contained in the set. Hence,
	$$0\leq \psi(y)\leq \psi(x)+\inner{\nabla\psi(x),y-x} = \psi(x) - r\|\nabla\psi(x)\|_x^*.$$
	Statement follows because $r$ is arbitrary in $[0,1)$.
\end{proof}

\section*{Acknowledgements}
We gratefully acknowledge the support of NSF under grants CAREER DMS-0954737 and CCF-1116928, as well as Dean's Research Fund.

\bibliographystyle{plain}
\bibliography{paper}

\begin{thebibliography}{10}

\bibitem{AbeAgaBarRak09}
J.~Abernethy, A.~Agarwal, P.~Bartlett, and A.~Rakhlin.
\newblock A stochastic view of optimal regret through minimax duality.
\newblock In {\em Proceedings of the 22nd Annual Conference on Learning
  Theory}, 2009.

\bibitem{AbeRak09colt}
J.~Abernethy and A.~Rakhlin.
\newblock Beating the adaptive bandit with high probability.
\newblock In {\em COLT}, 2009.

\bibitem{cesa1999minimax}
N.~Cesa-Bianchi and G.~Lugosi.
\newblock Minimax regret under log loss for general classes of experts.
\newblock In {\em Proceedings of the Twelfth annual conference on computational
  learning theory}, pages 12--18. ACM, 1999.

\bibitem{PLG}
N.~Cesa-Bianchi and G.~Lugosi.
\newblock {\em Prediction, Learning, and Games}.
\newblock Cambridge University Press, 2006.

\bibitem{freund1996predicting}
Y.~Freund.
\newblock Predicting a binary sequence almost as well as the optimal biased
  coin.
\newblock In {\em Proceedings of the ninth annual conference on Computational
  learning theory}, pages 89--98. ACM, 1996.

\bibitem{hazan2007logarithmic}
E.~Hazan, A.~Agarwal, and S~Kale.
\newblock Logarithmic regret algorithms for online convex optimization.
\newblock {\em Machine Learning}, 69(2-3):169--192, 2007.

\bibitem{mf-up-98}
N.~Merhav and M.~Feder.
\newblock Universal prediction.
\newblock {\em IEEE Transactions on Information Theory}, 44:2124--2147, 1998.

\bibitem{nesterov2011barrier}
Y.~Nesterov.
\newblock Barrier subgradient method.
\newblock {\em Mathematical programming}, 127(1):31--56, 2011.

\bibitem{rakhlin2012relax}
A.~Rakhlin, O.~Shamir, and K.~Sridharan.
\newblock Relax and randomize: From value to algorithms.
\newblock In {\em Advances in Neural Information Processing Systems 25}, pages
  2150--2158, 2012.

\bibitem{StatNotes2012}
A.~Rakhlin and K.~Sridharan.
\newblock Statistical learning and sequential prediction, 2012.
\newblock Available at {\small
  \url{http://stat.wharton.upenn.edu/~rakhlin/courses/stat928/stat928_notes.pdf}}.

\bibitem{RakSri14a}
A.~Rakhlin and K.~Sridharan.
\newblock Online nonparametric regression.
\newblock In {\em Conference on Learning Theory}, 2014.

\bibitem{RakSri15nonparametric}
A.~Rakhlin and K.~Sridharan.
\newblock Online nonparametric regression with general loss functions, 2015.
\newblock Available at \url{http://arxiv.org/abs/1501.06598}.

\bibitem{RakSriTew10}
A.~Rakhlin, K.~Sridharan, and A.~Tewari.
\newblock Online learning: Random averages, combinatorial parameters, and
  learnability.
\newblock {\em Advances in Neural Information Processing Systems 23}, pages
  1984--1992, 2010.

\bibitem{RakSriTew14ptrf}
A.~Rakhlin, K.~Sridharan, and A.~Tewari.
\newblock Sequential complexities and uniform martingale laws of large numbers.
\newblock {\em Probability Theory and Related Fields}, February 2014.

\bibitem{rissanen1986complexity}
J.~Rissanen.
\newblock Complexity of strings in the class of markov sources.
\newblock {\em Information Theory, IEEE Transactions on}, 32(4):526--532, 1986.

\bibitem{rissanen1996fisher}
J.~Rissanen.
\newblock Fisher information and stochastic complexity.
\newblock {\em Information Theory, IEEE Transactions on}, 42(1):40--47, 1996.

\bibitem{shtar1987universal}
Y.~M. Shtarkov.
\newblock Universal sequential coding of single messages.
\newblock {\em Problemy Peredachi Informatsii}, 23(3):3--17, 1987.

\bibitem{xie2000asymptotic}
Q.~Xie and A.R. Barron.
\newblock Asymptotic minimax regret for data compression, gambling, and
  prediction.
\newblock {\em Information Theory, IEEE Transactions on}, 46(2):431--445, 2000.

\end{thebibliography}

\end{document}